\documentclass[conference]{IEEEtran}
\usepackage[utf8]{inputenc}
\usepackage{amsmath}
\usepackage{amsfonts}
\usepackage{amssymb}
\usepackage{mathtools}
\usepackage{amsthm}
\usepackage{graphicx}
\usepackage{bm}
\usepackage[caption=false, font=footnotesize]{subfig}
\usepackage{xcolor}

\newtheorem{proposition}{Proposition}
\newtheorem{corollary}{Corollary}
\newtheorem{remark}{Remark}
\newtheorem{definition}{Definition}

\newcommand{\E}{\operatorname{E}}
\newcommand{\I}{\operatorname{I}}
\newcommand{\h}{\operatorname{h}}

\newcommand{\pr}{\operatorname{P}}

\renewcommand{\IEEEQED}{\IEEEQEDopen}
\begin{document}

\title{An Information-Theoretic Metric for Semantic Value of Spatiotemporal Information}
\author{
    \IEEEauthorblockN{
        Zijing Wang\IEEEauthorrefmark{1}, Zhijin Qin\IEEEauthorrefmark{2}, Siyu Lin\IEEEauthorrefmark{1}, and Wei Feng\IEEEauthorrefmark{2}
    }
    \IEEEauthorblockA{
        \IEEEauthorrefmark{1}School of Electronic and Information Engineering, Beijing Jiaotong University, Beijing 100044, China
    }
    \IEEEauthorblockA{
         \IEEEauthorrefmark{2}Department of Electronic Engineering, Tsinghua University, Beijing 100084, China\\
        }
        \IEEEauthorblockA{
        e-mail: zjwang1@bjtu.edu.cn; qinzhijin@tsinghua.edu.cn; sylin@bjtu.edu.cn; fengwei@tsinghua.edu.cn 
    }
}
\maketitle

\begin{abstract}
With the explosive growth of network scale and data volume, wireless communication is facing an increasingly severe limitation of spectrum resources. Semantic communication has emerged as a promising paradigm to break the bandwidth bottleneck by transmitting significant task-oriented semantic information rather than raw data. In practical real-time wireless applications, semantics of information exhibit diverse spatial and temporal correlations depending on intrinsic dynamics of source and extrinsic dynamics of environment. Motivated by this observation, this paper develops a novel information-theoretic metric to quantify the semantic value of spatiotemporal information. Specifically, a semantic value of information (SVoI) framework is proposed based on the mutual information, which characterises the reduction in uncertainty when predicting an unknown system state using past semantic spatiotemporal correlated observations. Focusing on general Gaussian Markov models, closed-form expressions of the SVoI are derived. Effects of both separable and coupled spatiotemporal correlations on SVoI are further investigated analytically. Numerical simulations are conducted to validate the theoretical analysis of SVoI and its bounds. The proposed SVoI metric jointly captures the impact of semantic spatiotemporal correlation of source, timeliness of information, and channel conditions, which could serve as an effective optimisation objective for the design of next-generation semantic-aware communication systems.
\end{abstract}


\section{Introduction}

With the network scale and data volume growing explosively, wireless communication systems face severe challenge of limited spectrum resources. Semantic communication has emerged as a new paradigm, providing potentials to break the bandwidth bottleneck~\cite{ProcIEEESurvey}. Rather than accurately delivering bit sequences with minimal distortion, semantic communication aims to effectively acquire goal-relevant or important information for task execution, thereby reducing communication overhead significantly. While semantic communication is envisioned as a key enabler of next-generation wireless systems, information-theoretic basis of semantic communication remains largely unexplored, and there is currently no universally accepted metric to quantify the semantic information. 

Existing theoretical studies primarily measure the semantic information from the perspective of information content and meaning. Logical probability was first introduced to define semantic entropy~\cite{SemEntropyCarnap52} by replacing traditional probability to quantify meaning uncertainty in the field of natural language processing. Semantic channel capacity~\cite{SemChannelCapBao11} was further formulated by explicitly accounting for semantic ambiguity, and its relationship to traditional Shannon channel capacity was also discussed. The classic Shannon information theory framework was mapped to its semantic counterpart by assuming the existence of synonymous mapping~\cite{niu2024}, which is a function to characterise semantic equivalence among different information but with same meaning. The rate distortion tradeoff between unobserved intrinsic semantic features and extrinsic noisy observations was further studied in~\cite{RateDistortion23TCOM}. 

The above works largely focus on a single piece of message, while ignoring the spatiotemporal correlation and temporal evolution of information, which are intrinsic properties of many data sources. In time-sensitive wireless applications, such as autonomous driving, environmental surveillance and industrial Internet of Things, the semantic relevance of information is not only determined by isolated messages, but is fundamentally shaped by when and where data is generated and received. Therefore, the demand for a rigorous semantic metric has led to the evolution of several time or space-dependent frameworks to quantify the quality of semantic information.

The measurement of semantic information has been largely explored from the perspective of information timeliness. Age of incorrect information (AoII)~\cite{AoIIEnabler} was proposed as a new semantic metric by combining the traditional linear age of information (AoI)~\cite{AoI2012infocom} and a task-oriented utility function. Furthermore, some works moved towards information-theoretic and estimation theory-based interpretations. The value of information (VoI)~\cite{previousTIT} utilised mutual information to measure the reduction in uncertainty of current status given past observations. A semantic-aware cost of actuation error (CAE) was defined as the weighted estimation cost for real-time control applications~\cite{nicolasCAE}, in which the weights are assigned based on the semantic significance. Regarding spatial dimensions,~\cite{AOISpatiotemporal} studied the spatiotemporal analysis of AoI, where network interference and link interactions cause AoI to fluctuate across space and time. Both spatial and temporal correlations were considered in remote sensing image processing~\cite{RemoteSensingSpatiotemporal} in which semi-supervised architecture was designed for efficient semantic feature extraction.

Despite these advances, existing metrics primarily treat the data as a one-dimensional temporal stream. Although some works considered the spatiotemporal correlation in terms of semantic feature extraction, they lack an information-theoretic foundation that can analytically quantify the semantic value across varying spatiotemporal intensities. In practical physical processes, the underlying physical phenomenon is inherently spatiotemporal, while different processes exhibit varying degrees of spatial and temporal correlation depending on their intrinsic dynamics and environmental conditions. Therefore, we are motivated to fill this gap by developing a new metric to capture the spatiotemporal correlation nature of semantic information in the context of information theory.

In this paper, we introduce a semantic value of information (SVoI) metric based on mutual information, which could be interpreted as the reduction in uncertainty to predict the unknown status through its past semantic spatiotemporal correlated observations. The closed-form expression of SVoI as well as its bounds are derived in general Gaussian Markov models. Effects of both separable and coupled semantic spatiotemporal correlations on SVoI are further analysed in detail. The performance of the SVoI is verified through simulation results. The proposed SVoI framework provides an analytically tractable tool for semantic information measurement and is able to work as the objective to facilitate the design and optimisation of semantic communication systems.

The rest of this paper is organised as follows. The proposed SVoI framework as well as its bounds are presented in Section II. The SVoI in general Gaussian Markov models is studied in Section III. The effect of semantic spatiotemporal correlation on the SVoI is discussed in Section IV. Numerical results are provided in Section V. Section VI concludes this work.

\section{Semantic Value of Information Formulation}
\subsection{Network Setting}
We assume that multiple unmanned devices are randomly deployed in the area to perform a monitoring task. These devices sample data consistently to get the latest status of the target physical process and transmit the data to a remote destination for analysis and processing. For each single device, the data samples are temporally correlated. For any two different devices, their data samples are spatially correlated. In practice, the underlying physical process, such as temperature, soil moisture level, and remote sensing images, is inherently spatiotemporal, but could exhibit different level of spatial and temporal correlation. 

Without loss of generality, the target physical phenomenon is modelled as a spatiotemporal stochastic process $\{X_{s,t}\}$ in which random variables $s$ and $t$ represent the space and time dimensions, respectively. As shown in Fig.~\ref{fig: model}, from the perspective of transmitter, the data sampled by device $i$ at time $t_j$ is denoted as 
\begin{equation}
    X_{s_i,t_j}: \bf{S} \times \bf{T} \to \mathbb{R}, 
\end{equation}
in which ${s_i}$ denotes the sampling position of the data packet, involving the two-dimensional coordinates of device $i$. The data samples are assumed to be transmitted to a remote destination though a noisy wireless channel. 

\begin{figure}
    \centering
    \includegraphics[height=6.8cm]{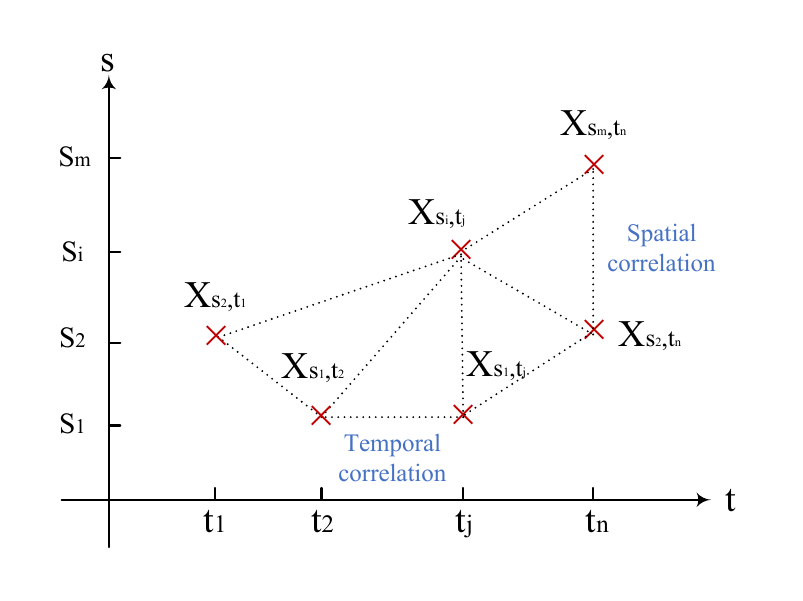}
    \caption{Spatiotemporal correlation model.}
    \label{fig: model}
\end{figure}

From the perspective of the receiver, the corresponding observation is denoted as $Y_{s_i,t'_j}$ where $t'_j$ is the receiving time. Due to the transmission delay, we have $t'_j \geq t_j$. Due to the negative impact of the transmission channel, the observation may be corrupted, i.e., $Y_{s_i,t'_j} \neq X_{s_i,t_j}$, but we have
\begin{equation}
    \pr[Y_{s_i,t'_j}\in \mathbb{A} | X_{s_i, t_1},\ldots,X_{s_i,t_j}] = \pr[Y_{s_i,t_j'}\in \mathbb{A} | X_{s_i,t_j}]
\end{equation}
for all device nodes $i$ and all admissible sets $\mathbb{A}$. The corruption in the observation could be caused by the channel fading, noise and interference from other devices. Given the current time $t$, the total number of observations is denoted as $n$. Therefore, data packets sampled by device $i$ and the observations at receiver during the time period $(0,t)$ can be given as $\{X_{s_i,t_1},\cdots,X_{s_i,t_n}\}$ and $\{Y_{s_i,t'_1},\cdots,Y_{s_i,t'_n}\}$ where $t_n \leq t'_n \leq t \leq t'_{n+1}$.


\subsection{Definition of SVoI}
The definition of the proposed SVoI metric is given as follows.
\begin{definition}
    The SVoI is defined as the mutual information between the unobserved status (at space $s$ and time $t$) $X_{s,t}$ of the transmitter and all the past observations $\bf{Y}^{m \times n}$ of the receiver, which is given as
\begin{equation}
\label{eq: SVoI definition}
    {\rm{SVoI}}(s,t) = \I({X_{s,t}}; {\bf{Y}}^{m \times n}),
\end{equation}
where matrix ${\bf{Y}}$ is given as
\begin{equation}
    {\bf{Y}}^{m \times n} = \left[ {\begin{array}{*{20}{c}}
{{Y_{{s_m},t{'_n}}}}&{{Y_{{s_m},t{'_{n - 1}}}}}& \cdots &{{Y_{{s_m},t{'_1}}}}\\
{{Y_{{s_{m - 1}},t{'_n}}}}&{{Y_{{s_{m - 1}},t{'_{n - 1}}}}}& \cdots &{{Y_{{s_{m - 1}},t{'_1}}}}\\
 \vdots & \vdots & \ddots & \vdots \\
{{Y_{{s_1},t{'_n}}}}&{{Y_{{s_1},t{'_{n - 1}}}}}& \cdots &{{Y_{{s_1},t{'_1}}}}
\end{array}} \right].
\end{equation}
\end{definition}
Here, ${\bf{Y}}$ is an $m\times n$-dimensional matrix, recording all the past observations. The corresponding data sample matrix at the source is denoted as
\begin{equation}
    {\bf{X}}^{m \times n} = \left[ {\begin{array}{*{20}{c}}
{{X_{{s_m},t{_n}}}}&{{X_{{s_m},t{_{n - 1}}}}}& \cdots &{{X_{{s_m},t{_1}}}}\\
{{X_{{s_{m - 1}},t{_n}}}}&{{X_{{s_{m - 1}},t{_{n - 1}}}}}& \cdots &{{X_{{s_{m - 1}},t{_1}}}}\\
 \vdots & \vdots & \ddots & \vdots \\
{{X_{{s_1},t{_n}}}}&{{X_{{s_1},t{_{n - 1}}}}}& \cdots &{{X_{{s_1},t{_1}}}}
\end{array}} \right].
\end{equation}
For notational convenience, we index the spatial observations from $m$ to $1$, where $i=m$ corresponds to the spatial location closest to the point of interest $s$. Similarly, the temporal observations are indexed from $n$ to $1$, where $j=n$ denotes the most recent time instant of the current time $t$. 

The SVoI is interpreted as the semantic information content of past observations to predict the unobserved status via their semantic spatiotemporal correlation, measuring how much those past observations reduce uncertainty about the target physical process at unknown status $(s,t)$.


\subsection{Bounds and Special Cases of SVoI}
Definition 1 gives a general and rigorous information-theoretic SVoI definition for arbitrary spatiotemporal processes. To gain more insight into the behaviour of SVoI, special cases are explored in this subsection to reveal the evolution and bounds of SVoI under different properties of spatiotemporal random processes.

The special cases considered include situations where the underlying random process $\{X_{s,t}\}$ is Markov, the observed process $\{Y_{s,t}\}$ exhibits Markov properties, and the observations are not affected by noise corruption. In this context, we have the following proposition.
\begin{proposition}
    The lower bound of the SVoI is given as
    \begin{equation}
    \label{eq:lower bound Y}
     {\rm{SVoI}}(s,t) \geq \I({X_{s,t}};{Y_{{s_m},t{'_n}}}),
\end{equation}
in which the equality holds if $X_{s,t} \perp\!\!\!\perp {\bf{Y}^{m \times n}} \,\big|\, {Y_{{s_m},t{'_n}}}$. An upper bound of the SVoI is given as
    \begin{equation}
    \label{eq:upper bound Xmn}
     {\rm{SVoI}}(s,t) \leq \I({X_{s,t}};{{\bf{X}}^{m \times n}}),
\end{equation}
in which the equality holds if there is no corruption in the observation at the receiver. When the underlying random process $\{X_{s,t}\}$ is a first-order Markov process in both time and space dimensions, the upper bound in~\eqref{eq:upper bound Xmn} turns out to be a tighter bound which is given as
\begin{equation}
    \label{eq:upper bound X }
     {\rm{SVoI}}(s,t) \leq \I({X_{s,t}};{X_{{s_m},{t_n}}}).
\end{equation}
\end{proposition}
\begin{proof}
    The lower bound in~\eqref{eq:lower bound Y} is obtained directly from the definition of mutual information with equality when we have $p(x_{s,t}|{\bf{y}}^{m\times n})=p(x_{s,t}|y_{{s_m},t{'_n}})$. In this case, the current unknown status of the underlying random process $X_{s,t}$ is conditionally independent of all the past spatiotemporal observations ${\bf{Y}}^{m \times n}$ given the latest and nearest observation ${Y_{{s_m},t{'_n}}}$. The upper bound in~\eqref{eq:upper bound Xmn}  is obtained by
        \begin{equation}
            \begin{aligned}
            {\rm{SVoI}}(s,t) &= \h({X_{s,t}}) - \h({X_{s,t}}|{{\bf{Y}}^{m \times n}})\\
 &\le \h({X_{s,t}}) - \h({X_{s,t}}|{{\bf{Y}}^{m \times n}},{{\bf{X}}^{m \times n}})\\
 &=\h({X_{s,t}}) - \h({X_{s,t}}|{{\bf{X}}^{m \times n}})\\
 &= \I({X_{s,t}};{{\bf{X}}^{m \times n}}).
        \end{aligned}  
        \end{equation}
\end{proof}
Proposition 1 illustrates insights into the evolution and limits of the SVoI in different cases. The lower bound represents the worst-case scenario, providing a baseline against which the performance of real-world systems can be compared. In contrast, upper bounds represent the ideal maximum achievable SVoI , serving as a theoretical benchmark for optimal system performance.

\section{SVoI under Gaussian Markov Models}

\subsection{Stationary Gaussian Markov Model}
The formulation and evolution of the SVoI provided in Section II accommodate arbitrary sets of spatiotemporal observations. In this section, we narrow the analysis in the context of Gaussian Markov spatiotemporal random processes to gain an intuitive understanding of the SVoI with its closed-form expression. This setup allows us to characterise the effect of spatiotemporal structure on SVoI under different Gaussian process kernels and to reveal the role of spatial and temporal correlation in shaping the semantic value. 

The underlying process is assumed to be a stationary Gaussian Markov process. Data samples generated at the source ${\bf{X}}^{m \times n}$ are jointly stationary Gaussian distributed. Due to the stationary nature, denote $\mu_x$ and $\sigma^2_x$ as the mean and variance, respectively, for any random variable $X_{s_i,t_j}$. For any two random variables $X_{s,t}$ and $X_{s_i,t_j}$, denote $\text{Cov}(X_{s,t},X_{s_i,t_j})$ as their covariance. Denote $\rho(d,\tau)$ as their autocorrelation coefficient
\begin{equation}
    \rho(d,\tau)=\frac{\text{Cov}(X_{s,t},X_{s_i,t_j})}{\sigma^2_x},
\end{equation}
where $d=||s-s_i||$ and $\tau=||t-t_j||$ due to the stationary property. The value of the autocorrelation coefficient ranges from $-1$ to $1$. If $\rho(d,\tau)$ takes $1$ (or $-1$), the correlation is perfect (or negatively perfect) in both space and time directions; $0$ means no linear correlation between the variables at the given space and time points.

These sampled data are assumed to be transmitted through a wireless channel with transmission noise. Denote ${\bf{N}}^{m \times n}$ as the noise matrix which includes all the noise variables $N_{s_i,t'_j}$ for all $i$ and $j$. The noise variables are independent and identically distributed (i.i.d.) Gaussian variables with zero mean and $\sigma^2_n$ variance, which are also independent of the data samples. In this context, the signal-to-noise ratio (SNR) could be written as
\begin{equation}
    \eta = \frac{\sigma^2_x}{\sigma^2_n}.
\end{equation}
From the perspective of the receiver, the observations are also jointly Gaussian distributed, which could be given as
\begin{equation}
     {Y_{s_i,t'_j}}= {X_{s_i,t_j}}+ {N_{s_i,t'_j}}, \quad 1 \le i \le m,   1 \le j \le n,
\end{equation}
and
\begin{equation}
     {\bf{Y}}^{m \times n}= {\bf{X}}^{m \times n}+ {\bf{N}}^{m \times n}.
\end{equation}

This assumption is widely used in practical applications, such as environmental and remote-sensing area. For example, soil moisture dynamics at temporal scales are commonly approximated by a discrete first-order autoregressive model or continuous Ornstein–Uhlenbeck (OU) model, which are typical stationary Gaussian Markov models. Similarly, spatial dependence is often captured by a Gaussian Markov random field model on a discretised grid, and the number of grid cells is $m$ in this work. The considered addictive white Gaussian noise channel could be used to model the ground-to-space channel in satellite communication networks.

\subsection{SVoI for the Gaussian Markov Model}
In the above stationary Gaussian Markov setting, the closed-form expression of SVoI in~\eqref{eq:lower bound Y} of Proposition 1 is derived with the following main result.
\begin{proposition}
    The lower bound of the SVoI in Gaussian Markov models could be given as
    \begin{equation}
        {\rm{SVoI}}(s,t) =  - \frac{1}{2}\log \left( {1 - \frac{{\eta\rho^2 (||s - {s_m}||,||t - {t_n}||)}}{{1 + \eta }}} \right).
    \end{equation}
\end{proposition}
\begin{proof}
   See the appendix.
\end{proof}

The lower bound of SVoI is considered as the exact SVoI involves high-dimensional mutual information with complex closed form. It is not just a simplification, but a practical way to capture the fundamental properties of the SVoI in a tractable form, enabling the practical system design by working as the optimisation goal. 

This proposition shows that the SVoI relates to the following two factors: the semantic spatiotemporal correlation property of the underlying source data (which is captured by the function $\rho(\cdot)$) and the SNR of the transmission channel (which is captured by the variable $\eta$). Their effect on the SVoI could be summarised as follows.

\begin{corollary}
    As the SNR turns to infinity, the lower bound of SVoI converges to the SVoI, i.e.,
    \begin{equation}
    \label{eq:large SNR}
         {\rm{SVoI}}(s,t) \to - \frac{1}{2}\log \left( {1 - {\rho ^2}(||s - {s_m}||,||t - {t_n}||)} \right).
    \end{equation}
As the correlation coefficient turns to $1$, the lower bound of the SVoI converges to
\begin{equation}
\label{eq: large rho}
     {\rm{SVoI}}(s,t) \to - \frac{1}{2}\log \left( {1 - \frac{{\eta}}{{1 + \eta }}} \right).
\end{equation}
As the SNR turns to be zero (or the correlation coefficient turns to $0$), we have ${\rm{SVoI}}(s,t) \to 0$.
\end{corollary}
\begin{proof}
    The results are obtained directly from Proposition 2, where we take $\eta \to \infty$, $\rho \to 1$, $\eta \to 0$ and $\rho \to 0$, respectively.
\end{proof}

Equation~\eqref{eq:large SNR} in this corollary verifies the special case given in Proposition 1 in which inequalities take the equal sign in high SNR regimes. In this special case, the observing process will be the same as the underlying process without the noise, thus the observing process will also be Markov. This means that the lower bound of SVoI is going to be tighter when the SNR $\eta$ goes larger. Equation~\eqref{eq: large rho} illustrates that, if the semantic spatiotemporal correlation is strong enough, the SVoI turns out to be independent regardless of space and time. 

\section{Effect of Semantic Spatiotemporal Correlation on SVoI}
Proposition 2 gives the general SVoI expression in Gaussian Markov models and shows that the spatiotemporal correlation property of the underlying data and the SNR of the wireless channel are the two dominant parts. Corollary 1 analyses the effect of SNR on the SVoI. This section considers two representative examples, further exploring the effect of spatiotemporal correlation on the semantic value in detail. 

\subsection{SVoI with Separable Spatiotemporal Correlation}
The first type is the separable spatiotemporal correlation. We take the Ornstein–Uhlenbeck kernel~\cite{ou1930OUmodel} as a representative example, which is able to characterise the exponential decay of information with respect to both temporal difference and spatial distance. The OU kernel is especially appropriate for physical processes in which spatial and temporal variations evolve smoothly and independently, such as modelling temperature fluctuations over a small geographic region. 

The autocorrelation coefficient between two variables with OU kernel is given as
\begin{equation}
\label{eq:OU rho}
     \rho_{\rm{sep}}(d,\tau)=e^{-\omega d}e^{-\kappa \tau},
\end{equation}
where $\omega>0$ and $\kappa>0$ are the spatial and temporal correlation decay rates, respectively. Small $\omega$ reveals strong semantic correlation in the space dimension, while large $\omega$ reveals weak semantic correlation. Similarly, large $\kappa$ reveals strong semantic correlation in the time dimension, while small $\kappa$ reveals weak semantic correlation.

\begin{corollary}
    Under the OU model where spatial and temporal dependencies are separable, the SVoI is given as 
\begin{equation}
\label{eq: OU+SVoI}
      {\rm{SVoI}}(s,t) =  - \frac{1}{2}\log \left( {1 - \frac{{\eta{e^{-2\omega{||s-s_m||}}}{e^{-2\kappa{(t-t_n)}}}}}{{1 + \eta }}} \right).
\end{equation}
\end{corollary}
\begin{proof}
    This result is obtained directly be letting~\eqref{eq:OU rho} into Proposition 2.
\end{proof}
Corollary 2 is established by the separable spatiotemporal model which demonstrates that when spatial and temporal correlations can be approximated as independent. In this case, the SVoI admits a concise closed-form expression which generalised the tradional AoI, i.e., $t-t_n$. This corollary shows that SVoI could work as the baseline that clearly characterises how the semantic value evolves with respect to time, space, and channel dynamics. 

\subsection{SVoI with Coupled Spatiotemporal Correlation}
In some practical scenarios, spatial and temporal dependencies are not fully separable but exhibit intricate coupling effects. For example, in environmental monitoring and climate modelling networks, spatial variability may influence temporal correlation structures, and vice versa. Therefore, Gneiting kernel~\cite{gneiting2002nonseparable} is further considered, providing a more flexible and realistic characterisation of spatiotemporal processes in which spatial and temporal dependencies are not independent but inherently coupled. 

The autocorrelation coefficient between two variables with Gneiting kernel is given as
\begin{equation}
\label{eq:Gneiting rho}
\rho_{\rm{coup}}(d,\tau)
=\frac{1}{\big(1+\alpha\,\tau^{u}\big)^{\varphi}}
\exp\!\left(-\,\frac{\beta\,d^{v}}{\big(1+\alpha\,\tau^{u}\big)^{\theta v}}\right).
\end{equation}
Here, $\alpha$ and $\beta$ represent the temporal and spatial scales, respectively; $u, v \in (0,2]$ control the temporal and spatial smoothness; $\varphi >0$ is the temporal decay exponent; $\theta \in [0,1]$ is the space–time coupling parameter in which $\theta =0$ yields the separable special case.

\begin{corollary}
    Under the Gneiting model where spatial and temporal dependencies are coupled, the SVoI is given as 
\begin{equation}
\label{eq: Gneiting+SVoI}
      {\rm{SVoI}}(s,t) =  - \frac{1}{2}\log \left( {1 -\frac{\eta}{1+\eta}\cdot \frac{{ e^{
-\,\frac{2\beta\,{||s-s_m||}^{v}}{(1+\alpha\,{(t-t_n)}^{u})^{\theta v}}}}}{{{\big(1+\alpha\,({t-t_n})^{u}\big)^{2\varphi}} }}} \right).
\end{equation}
\end{corollary}
\begin{proof}
    This result is obtained directly be letting~\eqref{eq:Gneiting rho} into Proposition 2.
\end{proof}

Compared with the two corollaries, we have the following observation.
\begin{remark}
    The separable case in Corollary 2 can be obtained as a special case of its coupled counterpart in Corollary 3.
\end{remark}
\begin{proof}
    When the coupling parameters are set to $\theta =0$, $v=1$, $\varphi=1$, $u=1$, the autocorrelation coefficient is given as
    \begin{equation}
    \begin{aligned}
         \rho_{\rm{coup}}(d,\tau)&=\frac{1}{(1+\alpha\,\tau)} \exp\!\left(-\,{\beta\,d}\right)\\
&\approx e^{-\alpha \tau} e^{-\beta d}\\
&=\rho_{\rm{sep}}(d,\tau).
    \end{aligned}
    \end{equation}
In this case, the coupled Gneiting kernel in Corollary 3 reduces to the separable OU form in Corollary 2.
\end{proof}

The study of these two representative classes of covariance functions is complementary. In both cases, the closed-form expression of the SVoI can be directly expressed as a function of the spatial variable (i.e., distance), the temporal variable (i.e., AoI), the channel variable (i.e., SNR), and the correlation parameters. The OU model provides a simple and interpretable baseline, in which spatial and temporal effects are independent, whereas the Gneiting model can capture the coupled interactions observed in many physical processes.

\section{Numerical Results}
In this section, numerical results are presented to evaluate the analytical results derived in Section III and IV and to illustrate the behaviour of the proposed SVoI metric. $40$ unmanned devices are distributed randomly in a $30 \times 30$ area. The monitoring point is set as the centre of the square area. The communication process is modelled as a M/M/1 queue. For each device, the sampling time is followed by a Poission process with rate $\lambda = 1$. For each data sample, the service time is modelled as an exponential process with rate $\mu = 2$. The total time slots $T_{\max}$ is $50$. For the case with coupled spatiotemporal correlation, the parameters $u,v,\varphi,\alpha,\beta$ are fixed and all set as $1$. Monte Carlo simulations are conducted in which the total number of running rounds is $10^5$.

\begin{figure}
    \centering
    \includegraphics[height=6.6cm]{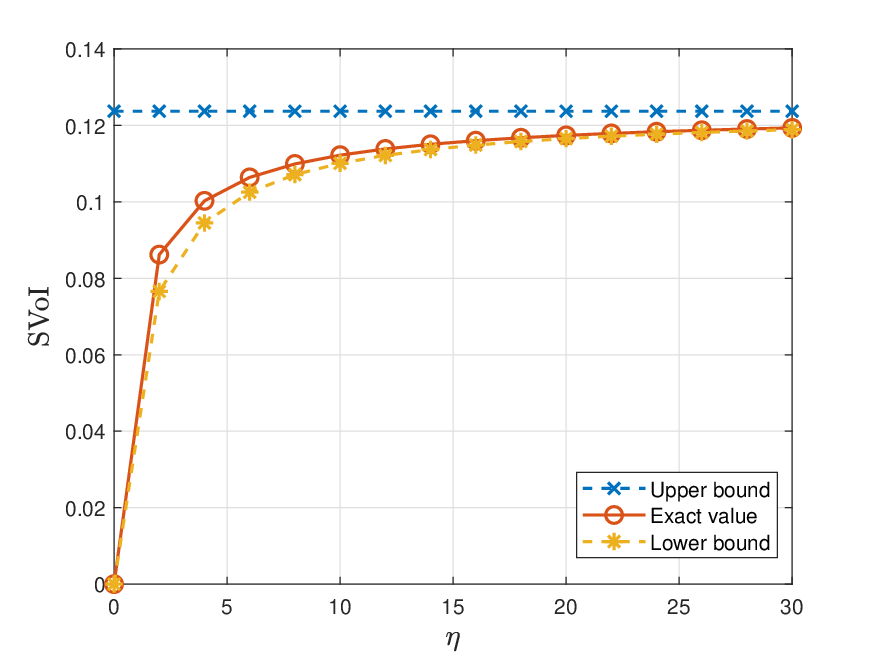}
    \caption{SVoI and its upper and lower bounds versus SNR.}
    \label{fig: bound}
\end{figure}

\begin{figure}
    \centering
    \includegraphics[height=6.6cm]{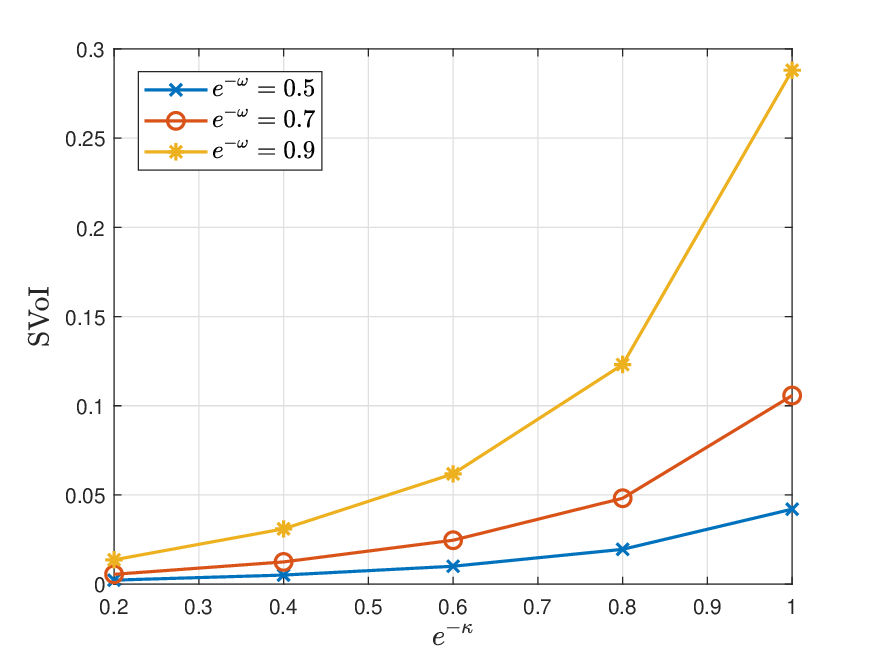}
    \caption{Separable spatiotemporal correlation: SVoI versus temporal correlation for different levels of spatial correlation. }
    \label{fig: sep}
\end{figure}

Fig.~\ref{fig: bound} shows the SVoI as well as its upper and lower bounds versus the SNR $\eta$. The bounds are obtained by Proposition 1 in which $m$ and $n$ are set as $1$ and $2$, respectively. As the SNR increases, the lower bound approaches the exact SVoI, and both of them converge to the upper bound. This behaviour is consistent with the analysis in Corollary 1 and indicates that the lower bound provides an accurate approximation of the exact SVoI, especially under high SNR conditions. This figure verifies the tightness of the bounds in Proposition 1.

Fig.~\ref{fig: sep} considers the separable spatiotemporal model and shows the performance of SVoI versus the temporal correlation coefficient $\kappa$ for different values of spatial correlation $\omega$. The SVoI increases monotonically with the intensity of temporal correlation. Similarly, stronger spatial correlation leads to a higher SVoI for a given temporal correlation. These results validate Corollary 2, demonstrating that different levels of spatial and temporal correlations yield different levels of semantic value.

Fig.~\ref{fig: coup} considers the coupled spatiotemporal model and shows the performance of SVoI versus SNR $\eta$ for different values of coupled spatiotemporal correlation intensity $\theta$. It is observed that the SVoI increases and converges with the SNR regardless of the spatiotemporal correlation. Stronger coupling yields higher SVoI at a given SNR. This observation aligns with Corollary 3 and highlights the impact of non-separable spatiotemporal dependencies on semantic information.

\begin{figure}
    \centering
    \includegraphics[height=6.6cm]{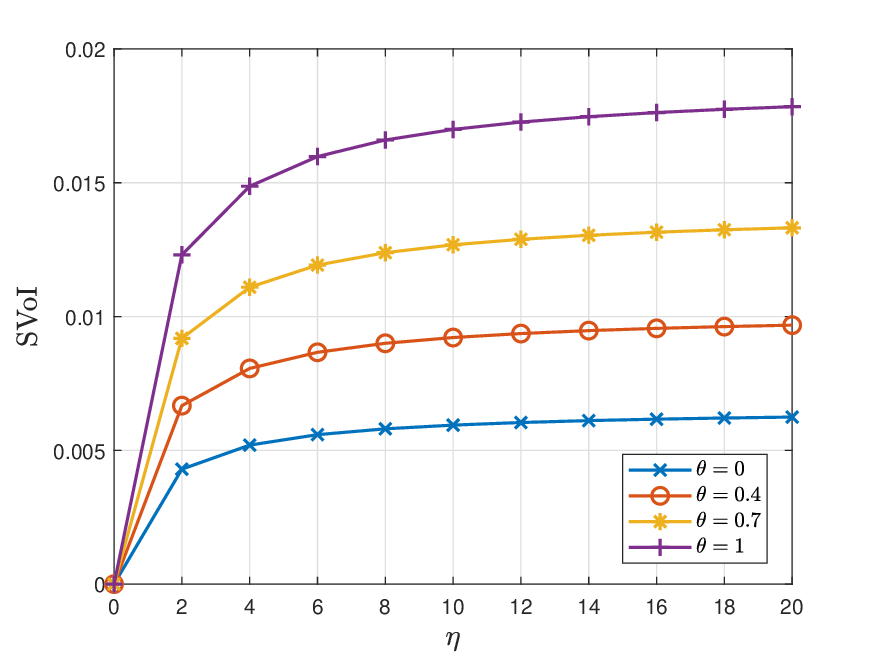}
    \caption{Coupled spatiotemporal correlation: SVoI versus SNR for different levels of spatiotemporal correlation. }
    \label{fig: coup}
\end{figure}

\section{Conclusions}
This paper proposed an information-theoretic metric named SVoI to quantify the semantic value of spatiotemporal information. The SVoI measures the amount of semantic information contained in past observations to predict the unknown status via their semantic spatiotemporal correlation. Closed-form expressions and tractable bounds of the SVoI were derived in Gaussian Markov models. Numerical results verified the theoretical analysis and demonstrated the tightness and effectiveness of the SVoI and its bounds. The proposed SVoI can capture the combined effects of SNR and spatiotemporal correlation properties, making it suitable for the performance evaluation of semantic communication systems.

\section*{Acknowledgment}
This work was supported in part by the National Natural Science Foundation of China (NSFC) under Grants 62401321, 62293484 and 62425110, and in part by Beijing Natural Science Foundation under Grant JR25020.



\section*{Appendix}
Variables ${X_{s,t}}$ and ${Y_{{s_m},t{'_n}}}$ are jointly Gaussian distributed, thus Proposition 2 follows from Proposition 1 in which the SVoI is lower bounded by
    \begin{multline}
        \I({X_{s,t}},{Y_{{s_m},t{'_n}}}) = \frac{1}{2}\log \left( {\frac{{\sigma _x^2\sigma _y^2}}{{|{\bf{\Sigma _{x,y}}}|}}} \right) \\
        = \frac{1}{2}\log \left( {\frac{1}{{1 - \gamma _{\text{XY}}^2(s,t;{s_m},t{'_n})}}} \right).
    \end{multline}
    Here, ${\bf{\Sigma _{x,y}}}$ is the covariance matrix of ${X_{s,t}}$ and ${Y_{{s_m},t{'_n}}}$ and $|\cdot|$ is the determinant operator. $\gamma_{\text{XY}}(\cdot)$ is the correlation coefficient which is written as
\begin{equation}
    \begin{aligned}
        {\gamma _\text{XY}}(s,t;{s_m},t{'_n}) &= \frac{{\E\left[ {({X_{s,t}} - {\mu _x})({Y_{{s_m},t{'_n}}} - {\mu _y})} \right]}}{{{\sigma _x}{\sigma _y}}}\\
        &=  \frac{{\E\left[ {({X_{s,t}} - {\mu _x})({X_{{s_m},{t_n}}} + {N_{{s_m},t{'_n}}} - {\mu _x})} \right]}}{{{\sigma _x}\sqrt {\sigma _x^2 + \sigma _n^2} }}\\
        & = \frac{{E\left[ {{X_{s,t}}{X_{{s_m},{t_n}}}} \right] - \mu _x^2}}{{\sigma _x^2\sqrt {1 + \frac{1}{\eta }} }}\\
        & = \frac{{\rho (||s - {s_m}||,||t - {t_n}||)}}{{\sqrt {1 + \frac{1}{\eta }} }}
    \end{aligned}
\end{equation}

\bibliographystyle{IEEEtran}
\bibliography{reference.bib}

@ARTICLE{RateDistortion23TCOM,
  author={Stavrou, Photios A. and Kountouris, Marios},
  journal={IEEE Trans. Commun.}, 
  title={The Role of Fidelity in Goal-Oriented Semantic Communication: {a} Rate Distortion Approach}, 
month={Jul.},
  year={2023},
  volume={71},
  number={7},
  pages={3918-3931},
  doi={10.1109/TCOMM.2023.3274122}}

@ARTICLE{AOISpatiotemporal,
  author={Yang, Howard H. and Arafa, Ahmed and Quek, Tony Q. S. and Poor, H. Vincent},
  journal={IEEE Transactions on Wireless Communications}, 
  title={Spatiotemporal Analysis for Age of Information in Random Access Networks Under Last-Come First-Serve With Replacement Protocol}, 
  year={2022},
  volume={21},
  number={4},
  pages={2813-2829},
  doi={10.1109/TWC.2021.3116041}}

@ARTICLE{RemoteSensingSpatiotemporal,
  author={Zhu, Linye and Xing, Huaqiao and Sun, Wenbin and Du, Shouhang and Fan, Deqin},
  journal={IEEE Transactions on Geoscience and Remote Sensing}, 
  title={Semi-Supervised Semantic Remote Sensing Image Change Detection Using Multimodal Spatiotemporal Association Knowledge}, 
  year={2025},
  volume={63},
  number={},
  pages={1-20},
  doi={10.1109/TGRS.2025.3621732}}

@ARTICLE{nicolasCAE,
  author={Luo, Jiping and Pappas, Nikolaos},
  journal={IEEE Transactions on Communications}, 
  title={Semantic-Aware Remote Estimation of Multiple Markov Sources Under Constraints}, 
  year={2025},
  volume={73},
  number={11},
  pages={11093-11105},
  doi={10.1109/TCOMM.2025.3571950}}

@misc{niu2024,
      title={A Mathematical Theory of Semantic Communication}, 
      author={Kai Niu and Ping Zhang},
      year={2024},
      eprint={2401.13387},
      archivePrefix={arXiv},
      primaryClass={cs.IT},
      url={https://arxiv.org/abs/2401.13387}, 
}

@book{SemEntropyCarnap52,
  title={An outline of a theory of semantic information},
  author={Carnap, Rudolf and Bar-Hillel, Yehoshua and others},
  year={1952},
  publisher={Research Laboratory of Electronics, Massachusetts Institute of Technology}
}

@article{ou1930OUmodel,
  title={On the Theory of the {Brownian} Motion {II}},
  author={Uhlenbeck, G. E. and Ornstein, L. S.},
  journal={Physical Review},
  volume={36},
  number={5},
  pages={823--841},
  year={1930},
  doi={10.1103/PhysRev.36.823}
}

@INPROCEEDINGS{SemChannelCapBao11,
  author={Bao, Jie and Basu, Prithwish and Dean, Mike and Partridge, Craig and Swami, Ananthram and Leland, Will and Hendler, James A.},
  booktitle={Proc. IEEE Netw. Science Workshop}, 
  title={Towards a theory of semantic communication}, 
month={Jun.},
  year={2011},
  volume={},
  number={},
  pages={110-117},
  doi={10.1109/NSW.2011.6004632}}

@article{gneiting2002nonseparable,
  title={Nonseparable, stationary covariance functions for space--time data},
  author={Gneiting, Tilmann},
  journal={Journal of the American Statistical Association},
  volume={97},
  number={458},
  pages={590--600},
  year={2002},
  publisher={Taylor \& Francis}
}

@ARTICLE{ProcIEEESurvey,
  author={Qin, Zhijin and Liang, Le and Wang, Zijing and Jin, Shi and Tao, Xiaoming and Tong, Wen and Li, Geoffrey Ye},
  journal={Proc. IEEE}, 
  title={{AI} Empowered Wireless Communications: {From} Bits to Semantics}, 
 year={2024},
  volume={112},
  number={7},
  pages={621-652},
  doi={10.1109/JPROC.2024.3437730}}

@INPROCEEDINGS{AoI2012infocom,
author={S. {Kaul} and R. {Yates} and M. {Gruteser}},
booktitle={Proc. IEEE INFOCOM},
title={Real-time status: {how} often should one update?},
year={2012},
volume={},
number={},
pages={2731-2735},}

@ARTICLE{previousTIT,
author={Wang, Zijing and Badiu, Mihai-Alin and Coon, Justin P.},
journal={IEEE Trans. Inf. Theory}, 
title={A framework for characterizing the value of information in hidden {M}arkov models},
year={2022},
volume={68},
number={8},
pages={5203-5216},
doi={10.1109/TIT.2022.3167545}}

@ARTICLE{AoIIEnabler,
  author={Maatouk, Ali and Assaad, Mohamad and Ephremides, Anthony},
  journal={IEEE Trans. Wireless Commun.}, 
  title={The Age of Incorrect Information: {An} Enabler of Semantics-Empowered Communication}, 
  year={2023},
  volume={22},
  number={4},
  pages={2621-2635},
  doi={10.1109/TWC.2022.3213227}}

\end{document}